\documentclass[dvipdfmx,10pt,conference,letter]{IEEEtran}
\usepackage{bm}                                                                                                         
\usepackage{graphicx}
\usepackage{amssymb}
\usepackage{amssymb}
\usepackage{amsxtra}
\usepackage{amsmath}
\usepackage{latexsym}
\usepackage{graphics}
\usepackage{verbatim}
\usepackage{epsfig}
\usepackage{setspace}
\usepackage{theorem}
\usepackage{rotating}
\usepackage{color}
\usepackage{cite}

\theorembodyfont{\normalfont\slshape}

\newtheorem{theorem}{Theorem}
\newtheorem{lemma}[theorem]{Lemma}

\def\ba{{\bm a}}
\def\bb{{\bm b}}
\def\bc{{\bm c}}
\def\bd{{\bm d}}
\def\be{{\bm e}}
\def\bf{{\bm f}}
\def\bg{{\bm g}}
\def\bh{{\bm h}}
\def\bp{{\bm p}}
\def\bq{{\bm q}}
\def\br{{\bm r}}
\def\bs{{\bm s}}
\def\bt{{\bm t}}
\def\bu{{\bm u}}
\def\bv{{\bm v}}
\def\bw{{\bm w}}
\def\bx{{\bm x}}

\def\cD{{\mathcal D}}

\def\cX{{\mathcal X}}

\def\DEL#1{}

\begin{document}

\thispagestyle{empty}
\title{Two-Dimensional Source Coding by Means of Subblock Enumeration}

\author{
\authorblockN{Takahiro Ota}
\authorblockA{Dept.\ of Computer \& Systems Engineering\\
Nagano Prefectural Institute of Technology\\
813-8, Shimonogo, Ueda, Nagano, 386-1211, JAPAN\\
Email: \texttt{ota@pit-nagano.ac.jp}}
\and 
\authorblockN{Hiroyoshi Morita}
\authorblockA{Graduate School of Informatics and Engineering\\
The University of Electro-Communications\\
1-5-1, Chofugaoka, Chofu, Tokyo, 182-8585, JAPAN\\
Email: \texttt{morita@uec.ac.jp}}
}

\maketitle

%\renewcommand{\thefootnote}{$*$}
%\footnotetext{This work was supported in part by a  Sience Finance of Ireland (SFI) of Ireland.}

\renewcommand{\thefootnote}{\arabic{footnote}}
\setcounter{footnote}{0}

\begin{abstract}
A technique of lossless compression via substring enumeration~(CSE)
attains compression ratios as well as popular lossless compressors
for one-dimensional (1D) sources.
The CSE utilizes a probabilistic model built from the circular string of an input source
for encoding the source.
The CSE is applicable to two-dimensional (2D) sources such as images by dealing with
a line of pixels of 2D source as a symbol of an extended alphabet.
At the initial step of the CSE encoding process, we need to output the number of occurrences
of all symbols of the extended alphabet, so that the time complexity increase exponentially
when the size of source becomes large.
To reduce the time complexity, we propose a new CSE which can encode a 2D source in
block-by-block instead of line-by-line. The proposed CSE utilizes the flat torus of an input 2D
source as a probabilistic model for encoding the source instead of the circular string of the source.
Moreover, we analyze the limit of the average codeword length of the proposed CSE for general sources.
\end{abstract}

\section{Introduction\label{intro}}
In 2010, Dub\'{e} and Beaudoin proposed an efficient off-line data compression algorithm
for a binary source known as {\em Compression via Substring Enumeration}~(CSE)~\cite{DB10}.
In~\cite{Yo11}, Yokoo proposed a universal CSE algorithm for a binary source
and various versions of the CSE for a binary source have been proposed so far~\cite{DY11, BD14, KYYK16}.
It is reported that performance of the CSE~\cite{BD14} is as well as that of 
an efficient off-line data compression algorithm using the Burrows-Wheeler transformation~(BWT)~\cite{BW94}.
In~\cite{OM13}, it is proved that an encoder, which is a deterministic finite automaton, of the CSE and an encoder
without sinks of the antidictionary coding~\cite{CMRS00} are isomorphic for a binary source.
Moreover, an antidictionary coding proposed in~\cite{OM141} provided  
the first CSE for $q$-ary ($q\!>\!2$) alphabet sources as a byproduct.
Iwata and Arimura 
proposed the modified algorithm and evaluated the maximum redundancy rate of the CSE for the $k$th order Markov sources~\cite{IA15}.

For encoding an input source, the CSE utilizes a probabilistic model built from 
the circular string which is obtained by concatenating the first symbol to the last symbol of the source.
A probabilistic model of the circular string is also useful for the BWT and antidictionary coding~\cite{OM13, OM141}, and in~\cite{CFMP16},
it is shown that an antidictionary built from the circular string is useful for genome comparison such as deoxyribonucleic acid~(DNA).
However, for a 2D source such as an image, computational time of the CSE is exponential
with respect to line length since the CSE works in line-by-line.
The CSE deals with a line of 2D source as a symbol of an extended alphabet.
At the initial step of the CSE encoding process, the CSE needs to output frequencies of all symbols of the extended alphabet.

To reduce the computational time, we propose a new CSE for a 2D source which utilizes the flat torus of an input 2D source %,
%which is obtained by concatenating the most left-hand side column (resp. the top row) to the most right-hand side column (resp. the bottom row) of the source,
as a probabilistic model instead of the circular string of the source. 
In the initial step, the total number of output blocks is constant
%The time complexity of the proposed CSE is polynomial order
%with respect to size of the source 
since the new CSE works in block-by-block. Moreover, we evaluate the limit of
the average codeword length of the proposed algorithm for general sources.

\section{Basic Notations and Definitions}\label{Notations}
\subsection{Alphabet and Block}
Let $\cX$ be a finite source alphabet $\{0, 1, \dots, J\!-\!1\}$ and let $|\cX|$ be a cardinality of $\cX$,
that is $|\cX| = J$.
Let $\cX^{[m,n]}$ be the set of all $m\!\times\!n$ finite blocks
$\bp=(p_{(i,j)})_{1\le i\le m, 1 \le j\le n}$ over $\cX$, where $p_{(i,j)} \in \cX$
is the element of $\bp$ at $(i,j)$-coordinate. 
Furthermore, let $\cX^{[*,*]}$ be $\cup_{m,n\ge 0}\cX^{[m,n]}, $
where $\cX^{[m,n]}$ includes the \emph{empty block} $\lambda^{[m, n]}$ 
when at least one of $m$ and $n$ is $0$.
For convenience, $\cX^{[m, 0]}$ and $\cX^{[0, n]}$ are defined as $\{\lambda^{[m, 0]}\}$
and $\{\lambda^{[0, n]}\}$, respectively. 
For $\bp \in \cX^{[*,*]}$, let $|\bp|_r$ and $|\bp|_c$
be the length of row (the \emph{height}) and the length of column (the \emph{width}), respectively.   
For example, when $\cX = \{0, 1\}$, Fig.~\ref{fig:ex33s} illustrates $\bp \in \cX^{[3, 3]}$ where $|\bp|_r\!=\!|\bp|_c\!=\!3$.

\begin{figure}[htb]
	\begin{minipage}{0.2\hsize}
		\centering
			\includegraphics[clip, width=0.8\linewidth]{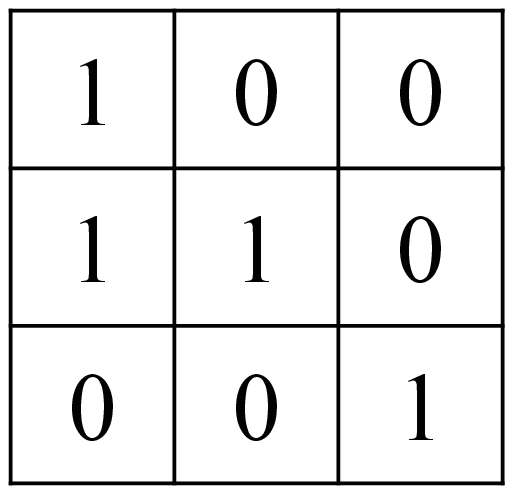}
		\caption{A $3 \times 3$ block $\bp$.}
		\label{fig:ex33s}
	\end{minipage}
	\begin{minipage}{0.75\hsize}
		\centering
			\includegraphics[clip, width=0.8\linewidth]{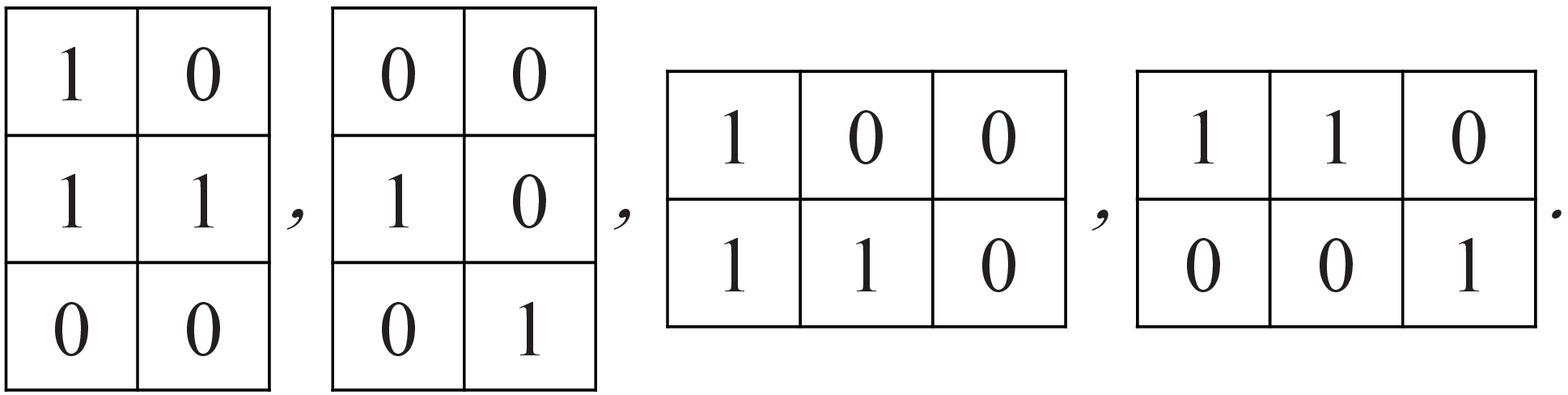}
		\caption{$\pi_c(\bp)$, $\sigma_c(\bp)$, $\pi_r(\bp)$, and $\sigma_r(\bp)$ of $\bp$ in Fig.~\ref{fig:ex33s}.}
		\label{fig:ps2}
	\end{minipage}
\end{figure}

\subsection{Subblock, Concatenation, and Dictionary}
For $\bp\!\in\!\cX^{[m, n]}$, a \emph{subblock} $\bp_{(i, j)}^{(i\!+\!k\!-\!1, j\!+\!l\!-\!1)}\!\!\!\in\!\!\cX^{[k, l]}$
is defined as 
\begin{equation*}
	\bp_{(i, j)}^{(i\!+\!k\!-\!1, j\!+\!l\!-\!1)}\!:=\!\left\{\begin{array}{ll}
				\lambda^{[0, l]}\ \,(k\!\le\!0\text{ and }l\!\ge\!0),&\nonumber\\
				\lambda^{[k, 0]}\ (k\!\ge\!0\text{ and }l\!\le\!0),&\nonumber\\
				\left(\begin{array}{ccc}
				p_{(i, j)} &\cdots &p_{(i, j\!+\!l\!-\!1)}\\
				\vdots&\ddots&\vdots\\
				p_{(i\!+\!k\!-\!1, j)}&\cdots &p_{(i\!+\!k\!-\!1, j\!+\!l\!-\!1)}
				\end{array}\right)&\\
				\ \ \ \ \ \ \ (k\!>\!0\text{ and }l\!>\!0)&\\

							\end{array} \label{eq:defxij}
				   \right.
\end{equation*}
where $1\!\leq\!i\!\leq\!m$, $1\!\leq\!j\!\leq\!n$, $k\!\leq\!m\!-\!i\!+\!1$, and $l\!\leq\!n\!-\!j\!+\!1$.
Hereinafter, without notice, we assume that the height and width of $\bp$ are respectively given by $m~(\ge 2)$ and $n~(\ge 2)$.
In particular, $(m-1)\times n$ subblocks $\bp_{(1, 1)}^{(m\!-\!1, n)}$ and $\bp_{(2, 1)}^{(m, n)}$ are
denoted by $\pi_r(\bp)$ and $\sigma_r(\bp)$, respectively.
Moreover, $m\times (n-1)$ subblocks $\bp_{(1, 1)}^{(m, n\!-\!1)}$ and $\bp_{(1, 2)}^{(m, n)}$
are denoted by $\pi_c(\bp)$ and $\sigma_c(\bp)$, respectively.
For example, for $\bp$ in Fig.~\ref{fig:ex33s}, 
Fig.~\ref{fig:ps2} shows $\pi_c(\bp)$, $\sigma_c(\bp)$, $\pi_r(\bp)$, and $\sigma_r(\bp)$
from the left-hand side. 

For $\bp$, the dictionary of $\bp$ is defined as the set of all the subblocks of $\bp$, that is,  
\begin{align*}
  \cD(\bp)\!:=\! 
    \{\bp_{(i, j)}^{(i\!+\!k\!-\!1, j\!+\!l\!-\!1)}\text{ s.t. }&1\!\leq\!i\!\leq\!m, 1\!\leq\!j\!\leq\!n, \nonumber\\
	&0\!\leq\!k\!\leq\!m\!-\!i\!+\!1, 0\!\leq\!l\!\leq\!n\!-\!j\!+\!1\}.
\end{align*}

Now we define a concatenation of blocks by column-wisely as follows:  
For two blocks $\bs, \bt\!\in\!\cX^{[*, *]}$ such that $|\bs|_r\!=\!|\bt|_r$, 
define $\bs\!:\!\bt \in \cX^{[|\bs|_r, |\bs|_c+|\bt|_c]}$ to be a block obtained by concatenating $\bt$ at the end of $\bs$ in columns.
Similarly, 
we define a concatenation of blocks by row-wisely as follows:  
for two blocks $\bu, \bv\!\in\!\cX^{[*, *]}$ such that $|\bu|_c\!=\!|\bv|_c$, 
define $\bu\!/\!\bv \in \cX^{[|\bu|_r+|\bv|_r, |\bu|_c]}$ to be a block obtained by concatenating $\bu$ at the end of $\bv$ in rows.

\subsection{Flat Torus, Primitive, and Frequencies of Subblocks}
For $\bp$, a \emph{flat torus} of $\bp$, denoted by $\bp^T$,
is constructed by concatenating the most left-hand side column (resp. the top row) 
to the most right-hand side column (resp. the bottom row) of $\bp$.
The flat torus can be treated as an infinite pattern such that $p_{(i, j)} = p^T_{(i+km, j+ln)}$ for non-negative integer $k, l$.

For $\bq \in \cX^{[m, n]}$ and $\bar{\bp}:=(\bp\!:\!\bp)/(\bp\!:\!\bp)$,  
if there exist positive integers $i~(1\leq i \leq m)$ and $j~(1\leq j \leq n)$ such that
$\bq = \bar{\bp}_{(i, j)}^{(i+m-1, j+n-1)}$ is satisfied, then the equivalence relation is denoted as
$\bq \simeq \bp$. Note that $\bar{\bp}$ is a $2m\times 2n$ subblock of $\bp^T$.
Let $[\bp]$ be the set of all the blocks $\bq$ such that $\bq \simeq \bp$,
\begin{align}
[\bp] := \{\bq \in \cX^{[m, n]}\text{ s.t. }\bq \in \mathcal{D}(\bar{\bp})\}.
\end{align}
If $|\,[\bp]\,| = mn$, $\bp$ is called \emph{primitive}. Hereinafter, without notice, we assume that $\bp$ is primitive. 
For example, $\bp$ shown in Fig.~\ref{fig:ex33s} is primitive.

For $\bp$ and $\bu \in \cX^{[k, l]}$~($0\!\leq\!k\!\leq\!m$ and $0\!\leq\!l\!\leq\!n)$,
\begin{align}
N(\bu\,|\,\bp)\!&:=\!|\,\{\br\text{ s.t. } \bu = \br_{(1, 1)}^{(k, l)}, \br \in [\bp]\}\,|\label{eq:N}
\end{align}
where $N(\lambda^{[k, l]}|\bp) = mn$~($k = 0$ or $l = 0$).
For convenience, we often adopt the notation $N(\bu)$ instead of $N(\bu|\bp)$.
For $\bp$, $0\!\leq\!k\!\leq\!m$, and $0\!\leq\!l\!\leq\!n$, 
\begin{align}
\sum_{{\bm u} \in \cX^{[k, l]}} N(\bu) = mn. \label{eq:CD1}
\end{align}
Moreover, 
for $\bv \in \cX^{[i, j]}~(0\!\leq\!i\!\leq\!m, 0\!\leq\!j\!<\!n)$ and $\bv' \in \cX^{[k, l]}~(0\!\leq\!k\!<\!m, 0\!\leq\!l\leq\!n)$,
\begin{align}
N(\bv)\!&=\!\!\!\!\!\!\!\sum_{\bc\in\cX^{[i, 1]}}\!\!\!\!N(\bc\!:\!\bv)\!=\!\!\!\!\!\!\!\sum_{\bc\in\cX^{[i, 1]}}\!\!\!\!N(\bv\!:\!\bc),\label{eq:CD2}\\
N(\bv')\!&=\!\!\!\!\!\!\!\sum_{\br\in\cX^{[1, l]}}\!\!\!\!N(\br/\bv')\!=\!\!\!\!\!\!\!\sum_{\br\in\cX^{[1, l]}}\!\!\!\!N(\bv'/\br).\label{eq:CD3}
\end{align}
\subsection{Classifications of Flat Tori and Core}
For $\bp$ and $k~(0 \leq k \leq m)$, and $l~(0 \leq l \leq n)$, 
\begin{align}
\mathcal{T}(\bp, k, l) &:= \{\bq \in \cX^{[m, n]}\text{ s.t. }N(\bw|\bq) = N(\bw|\bp),\nonumber\\
					   &\ \ \ \ \ \  {}^\forall \bw \in \cX^{[k, l]}, \bq\text{ is primitive.}\}
\end{align}
For example, $[\bp] = \mathcal{T}(\bp, m, n)$.
For $0\!\leq\!k\!<\!n$ and fixed $0\!\leq\!l\!\leq\!n$, 
$\mathcal{T}(\bp, k, l)$ is monotone decreasing with $k$, that is
$\mathcal{T}(\bp, k\!+\!1, l)\!\subset\!\mathcal{T}(\bp, k, l)$.
Similarly, 
for fixed $0\!\leq\!k'\!\leq\!n$ and $0\!\leq\!l'\!<\!n$, 
$\mathcal{T}(\bp, k', l'\!+\!1)\!\subset\!\mathcal{T}(\bp, k', l')$.
Next, we define $\mathcal{B}(\bp)$, 
\begin{align}
\mathcal{B}(\bp)\!:=\!\{\bb\!\in\!\cX^{[k, l]}\text{ s.t. }&\sigma_r(\pi_r(\bb))\!\in\!\mathcal{D}(\bar{\bp}), \sigma_c(\pi_c(\bb))\!\in\!\mathcal{D}(\bar{\bp}),\nonumber\\
&1\!\leq\!k\!\leq\!m, 1\!\leq\!l\!\leq\!n\}\cup\{\lambda^{[0, 0]}\}.
\end{align}

We assume that elements of $\mathcal{B}(\bp)$ are ordered
in ascending order with its height (if heights of the elements are equal, then the elements ordered with its width;
if widths of the elements are equal, then the elements are ordered in lexicographical order column-wisely)
where $\bb_i$ is the $i$th element of $\mathcal{B}(\bp)~(1\!\leq\!i\!\leq\!|\mathcal{B}(\bp)|)$.
For $i~(1\!\leq\!i\!\leq\!|\mathcal{B}(\bp)|)$, 
\begin{align}
\mathcal{T}(\mathcal{B}(\bp), \bp, i) := \{&\bq \in \cX^{[m, n]}\text{ s.t. }N(\bb_j\,|\,\bq)\!=\!N(\bb_j\,|\,\bp),\nonumber\\
&1 \leq {}^\forall j \leq i, \bq\text{ is primitive.}\}
\end{align}
For example, $[\bp]\!=\!\mathcal{T}(\mathcal{B}(\bp), \bp, |\mathcal{B}(\bp)|)$.
For $1\!\leq\!i\!<\!|\mathcal{B}(\bp)|$, 
$\mathcal{T}(\mathcal{B}(\bp), \bp, i)$ is monotone decreasing with $i$, that is
$\mathcal{T}(\mathcal{B}(\bp), \bp, i\!+\!1)\!\subset\!\mathcal{T}(\mathcal{B}(\bp), \bp, i)$.

A $\bu\!\in\!\mathcal{B}(\bp)$ such that  
%
%\begin{align}
$\ba:\bu, \bb:\bu, \bu:\bc, \bu:\bd\!\in\!\mathcal{D}(\bar{\bp})$
%\end{align}
%
where $\ba, \bb(\neq\!\ba), \bc, \bd(\neq\!\bc)\!\in\!\cX^{[|\bu|_r, 1]}$ is called \emph{c-core}.
A $\bv \in \mathcal{B}(\bp)$ such that  
%\begin{align}
$\be/\bv, \bf/\bv, \bv/\bg, \bv/\bh \in \mathcal{D}(\bar{\bp})$
%\end{align}
%
where $\be, \bf(\neq\!\be), \bg, \bh(\neq\!\bg)\!\in\!\cX^{[1, |\bv|_c]}$ is called \emph{r-core}.

\section{Review of Conventional CSE}~\label{sec:Conv}
The conventional CSE is a lossless compression algorithm for a 1D source.
For $\bp$, we can regard $\bp$ as a 1D source $\bx \in \hat{\cX}^{[1, n]}$ over an extended alphabet $\hat{\cX} (=\cX^{[m, 1]})$,
so that the CSE can encode $\bp$ as a 1D source $\bx$.
For $\bx$,
the CSE outputs a following triplet
\begin{align}
(E(n), e(\bb_2, \bb_3, \dots, \bb_{|\mathcal{B}(\bx)|}), \epsilon(\text{rank($\bx$)})). \label{eq:Ccode}
\end{align}
In (\ref{eq:Ccode}), $E(n)$ represents an encoded $n$ by means of Elias integer code~\cite{Eli75}.
And rank($\bx$) represents an index for identifying $\bx$ in $[\bx]$
such as the rank of $\bx$ in $[\bx]$ with lexicographical order.
Then, $\epsilon$(\text{rank($\bx$)}) represents an encoded rank($\bx$) by $\lceil \log_2 n \rceil$ bits, and
$e(\bb_2, \bb_3, \dots, \bb_{|\mathcal{B}(\bx)|})$
represents a sequence of $N(\bb_i)~(2\leq i \leq |\mathcal{B}(\bx)|)$ which
are encoded by an entropy coding where $N(\bb_i)$ represents $N(\bb_i|\bx)$ in this subsection.
In encoding, for $\bb_i \in \mathcal{B}(\bx)$, 
$i$ is selected from 2 to $|\mathcal{B}(\bx)|$ since $N(\bb_1) = N(\lambda^{[0, 0]}) = n$
and $n$ is encoded as $E(n)$. 
For $2 \leq i \leq |\mathcal{B}(\bx)|$,  
\begin{description}
\item[(C-i)] in case of $|\bb_i|_c\!=\!1$: Encode $N(\bb_i)$ if $\bb_i \neq \bb_{|\hat{\cX}|\!+\!1}$,
\item[(C-ii)]in\,case of $|\bb_i|_c\!\geq\!2$: Encode $N(\bb_i)$ if (\ref{eq:CIA}) holds and 
$\ba, \bc\!\in\!\hat{\cX}\!\backslash\{\bb_{|\hat{\cX}|\!+\!1}\}$
where $\bb_i = \ba\!:\!\bw\!:\!\bc$ such that $\bw = \sigma_c(\pi_c(\bb_i))$
\end{description}
where $\bb_{|\hat{\cX}|\!+\!1}$ is the element of $\hat{\cX}$ having the largest index in $\mathcal{B}(\bx)$
and note that (\ref{eq:CIA}) was first shown in \cite{IA15}.
Note that in (C-i), $N(\bb_i)$ is encoded even if $N(\bb_i) = 0$.

In (C-i), $N(\bb_{|\hat{\cX}|\!+\!1})$ can be calculated
by using (\ref{eq:CD1}) and already encoded $\bb_j (j < |\hat{\cX}|\!+\!1)$. Similarly, in (C-ii), $N(\bb_i)$ such that $\ba = \bb_{|\hat{\cX}|\!+\!1}$ or 
$\bc = \bb_{|\hat{\cX}|\!+\!1}$ can be calculated by using (\ref{eq:CD2}) and $\bb_k~(k < i)$. Therefore, they are not encoded.
\begin{align}
\min(N(\ba\!:\!\bw),\, &N(\bw\!:\!\bc), N(\bw)\!-\!N(\ba\!:\!\bw),\nonumber\\
&N(\bw)\!-\!N(\bw\!:\!\bc)) \geq 1. \label{eq:CIA}
\end{align}
As for $\bb_i(=\ba\!:\!\bw\!:\!\bc)$ in (C-ii), satisfying (\ref{eq:CIA}) is the same that $\bw$ is a c-core.
Moreover, since $\ba, \bw, \bc\!\in\!\mathcal{D}(\bar{\bx})$ and (\ref{eq:CD1}) holds, number of candidates of $\bb_i$ for encoding in (C-ii) is polynomial order with $n$.
The details are described in the bottom of this section.
In (C-i), $N(\bb_i)$ satisfies the following inequality
\begin{align}
0 \leq N(\bb_i) \leq n-1. \label{iq:Ci}
\end{align}
In (C-ii), $N(\bb_i)$ satisfies the following inequality~\cite{OM141}
\begin{align}
&\max\{0, N(\ba\!:\!\bw)-\!\!\!\!\!\!\!\sum_{\bd \in \hat{\cX}\backslash \{\bc\}}\!\!\!\!\!\!N(\bw\!:\!\bd), N(\bw\!:\!\bc)-\!\!\!\!\!\!\!\sum_{\bb \in \hat{\cX}\backslash \{\ba\}}\!\!\!\!\!\!N(\bb\!:\!\bw)\} \nonumber\\
&\leq N(\ba\!:\!\bw\!:\!\bc)\!\leq\!\min\{N(\ba\!:\!\bw), N(\bw\!:\!\bc)\}. \label{iq:Cii}
\end{align}
The left-hand side term in (\ref{eq:CIA}) is given by the difference between the 3rd term
and the 1st term in (\ref{iq:Cii}). %, that is
Therefore, if (\ref{eq:CIA}) does not hold, then the 1st and the 3rd terms are equal.
In other words, $N(\bb_i) = \min\{N(\ba\!:\!\bw), N(\bw\!:\!\bc)\}$ holds, so that $N(\bb_i)$ can be calculated.
Hence, $N(\bb_i)$ is not encoded if (\ref{eq:CIA}) does not hold.

Let $I(\ba\!:\!\bw\!:\!\bc)$ be $\min(N(\ba\!:\!\bw), N(\bw\!:\!\bc), N(\bw)\!-\!N(\ba\!:\!\bw),$
$N(\bw)\!-\!N(\bw\!:\!\bc)) + 1$
where $\min(\cdot)$ is the left-hand term of (\ref{eq:CIA}).
For encoding $N(\bb_i)$ by an entropy coding,
a probability is assigned to $N(\bb_i)$ as follows~\cite{Yo11}.
\begin{align}
\frac{1}{n}\ \ &(|\bb_i|_c = 1), \label{eq:PC0}\\
\frac{1}{I(\bb_i)} \ \ &(2 \leq |\bb_i| _c\leq \left\lfloor \log_2\log_2 n \right\rfloor),\label{eq:PC1}\\
\frac{|\mathcal{T}(\mathcal{B}(\bx), \bx, i)|}{|\mathcal{T}(\mathcal{B}(\bx), \bx, i\!-\!1)|}\ \ &(|\bb_i|_c \geq \left\lfloor \log_2\log_2 n \right\rfloor\!+\!1). \label{eq:PC2}
\end{align}
The assigned probabilities are encoded by an entropy coding such as an arithmetic coding~\cite{MT95}.

For encoding 2D source $\bp$ by the conventional CSE,
there is a problem with respect to computational time.
In (C-i), number of encoded $N(\bb_i)~(2\leq i \leq |\hat{\cX}|)$ is exponential with respect to $m$
since $|\hat{\cX}|$ is $|\cX|^m$. In practical, $m$ is greater than 1000 for an image $\bp \in \cX^{[m, n]}$,
so that the number is greater than $2^{1000}$ even if $|\cX| = 2$.
Note that in (C-ii), number of encoded $N(\bb_i)$ is not exponential with respect to $m$ and $n$.
The reason is as follows.
Since $\bw$ is a c-core, from (\ref{eq:CD1}) and (\ref{eq:CD2}), the total number of c-cores is polynomial order with respect to $m$ and $n$.
Moreover, since $N(\ba\bw)\ge 1$ and $N(\bw\bc)\ge 1$ in (\ref{eq:CIA}), $\ba, \bc \in \cD(\bar{\bx})\cap \hat{\cX}$ also hold.
From (\ref{eq:CD1}) and (\ref{eq:CD2}),
$|\cD(\bar{\bx}) \cap \hat{\cX}|$ never exceeds $mn$.
Hence, the total number of candidates $\bb_i (=\ba\!:\!\bw\!:\!\bc)$ for encoding in (C-ii) is polynomial order with respect to $m$ and $n$.
In other words, the set of all the candidates can be utilized instead of $\mathcal{B}(\bx)$ in (C-ii) in practice.
Note that $\mathcal{B}(\bx)$ is utilized for simplifying the explanation in this paper.
As for compression ratio, only a relation on column is utilized as shown in (\ref{eq:CIA})
and a relation on row is not utilized.
\section{Proposed Algorithm}~\label{Proposed}
For $\bp$, we assume that $m \le n$. %, and let $A$ be the aspect ratio $n/m$.
Let $K$ and $L$ be $\lfloor \sqrt{\log_{|\cX|}\log_{|\cX|} m} \rfloor$ and $\lfloor \sqrt{\log_{|\cX|}\log_{|\cX|} n} \rfloor$, respectively.

We divide $\mathcal{B}(\bp)$ into four disjoint parts with respect to size of its elements.
\begin{align*}
\mathcal{B}_0(\bp)\!&:=\!\{\bb\!\in\!\mathcal{B}(\bp)\text{ s.t. }\bb = \lambda^{[0, 0]}\},\\
\mathcal{B}_1(\bp)\!&:=\!\{\bb\!\in\!\mathcal{B}(\bp)\text{ s.t. }\bb \in \cX\},\\
\mathcal{B}_2(\bp)\!&:=\!\{\bb\!\in\!\mathcal{B}(\bp)\text{ s.t. }1\!\le\!|\bb|_r\!\le K,\!1\le\!|\bb|_c\!\le\!L, \bb\!\notin\!\cX\},\\ 
\mathcal{B}_3(\bp)\!&:=\!\{\bb\!\in\!\mathcal{B}(\bp)\text{ s.t. }K < |\bb|_r\text{ or }L < |\bb|_c\}.
\end{align*}
Elements of $\mathcal{B}_i(\bp)~(i = 0, 1, 2, 3)$ are ordered
in ascending order with its height (if heights of the elements are equal, then the elements ordered with its width;
if widths of the elements are equal, then the elements are ordered in lexicographical column-wisely.)
Then, elements of $\mathcal{B}(\bp)$ are reordered with $(\mathcal{B}_0(\bp), \mathcal{B}_1(\bp), \mathcal{B}_2(\bp), \mathcal{B}_3(\bp))$.
For $2 \leq i \leq |\mathcal{B}(\bp)|$,  
\begin{description}
\item[(P-i)] in case of $\bb_i\!\in\!\mathcal{B}_1(\bp)$: Encode $N(\bb_i)$ if $\bb_i \neq J\!-\!1$,
\item[(P-ii)]in\,case of $\bb_i\!\in\!\mathcal{B}_2(\bp)\!\cup\!\mathcal{B}_3(\bp)$:
\item[\ \ \ \ \ 1)] if $|\bb_i|_c\!=\!1$: Encode $N(\bb_i)$ if (\ref{eq:CIA}) holds and $\ba, \bc\!\in\!\cX\backslash\{J\!-\!1\}$
where $\bb_i\!=\!\ba\!\!:\!\!\bw\!\!:\!\!\bc$ such that $\bw\!=\!\sigma_c(\pi_c(\bb_i))$,
\item[\ \ \ \ \ 2)] if $|\bb_i|_r\!=\!1$: Encode $N(\bb_i)$ if (\ref{eq:CIA2}) holds and $\be, \bg\!\in\!\cX\backslash\{J\!-\!1\}$
where $\bb_i\!=\!\be\!/\!\bv\!/\!\bg$ such that $\bv\!=\!\sigma_r(\pi_r(\bb_i))$,
\item[\ \ \ \ \ 3)] if $|\bb_i|_c \ge 2$ and $|\bb_i|_r \ge 2$: Encode $N(\bb_i)$ if both (\ref{eq:CIA}) and (\ref{eq:CIA2}) hold
where $\ba, \bc\!\in\!\cX^{[|\bb_i|_r, 1]}\backslash\{\bx(|\bb_i|_r, 1)\}$ and $\be, \bg\!\in\!\cX^{[1, |\bb_i|_c]}\backslash\{\bx(1, |\bb_i|_c)\}$,
\end{description}
where $\bx(k, 1)$ and $\bx(1, l)$ are the element of $\cX^{[k, 1]}$ and $\cX^{[1, l]}$ having the largest index in $\mathcal{B}(\bp)$, respectively.
\begin{align}
\min(N(\be\!/\!\bv),\, &N(\bv\!/\!\bg), N(\bv)\!-\!N(\be\!/\!\bv),\nonumber\\
&N(\bv)\!-\!N(\bv\!/\!\bg)) \geq 1. \label{eq:CIA2}
\end{align}
As for $\bb_i(=\be\!/\!\bv\!/\!\bg)$ in 2) and 3), satisfying (\ref{eq:CIA2}) is the same that $\bv$ is a r-core.
As shown in the discussions in Sec.~\ref{sec:Conv}, 
number of candidates of $\bb_i$ for encoding in (P-ii) is polynomial order with $m$ and $n$.
The details are described in the bottom of this section.

The conventional CSE utilizes only condition (\ref{eq:CIA}) with respect to column,
while the proposed algorithm utilizes conditions (\ref{eq:CIA}) and (\ref{eq:CIA2}) with respect to column and row,
respectively, for encoding $\bp$. In 1) and 2), $\bb_i$ is one row and one column, so that (\ref{eq:CIA}) and ({\ref{eq:CIA2}) is only utilized, respectively. 
In (P-i), $N(\bb_i)$ satisfies $0\!\le\!N(\bb_i)\le\!mn\!-\!1$.
In (P-ii), $N(\bb_i)$ such that $|\bb_i|_c\ge 2$ satisfies a modified (\ref{iq:Cii}) which is obtained by replacing $\hat{\cX}$ by $\cX^{[|\ba|_r, 1]}$, and 
$N(\bb_i)$ such that $|\bb_i|_r \ge 2$ satisfies the following inequality
\begin{align}
&\max\{0, N(\be\!/\!\bv)-\!\!\!\!\!\!\!\sum_{\bh \in \cX^{[1, |\be|_c]}\backslash \{\bg\}}\!\!\!\!\!\!N(\bv\!/\!\bh), N(\bv\!/\!\bg)-\!\!\!\!\!\!\!\sum_{\bf \in ^{[1, |\be|_c]}\backslash \{\be\}}\!\!\!\!\!\!N(\bf\!/\!\bv)\} \nonumber\\
&\leq N(\be\!/\!\bv\!/\!\bg)\!\leq\!\min\{N(\be\!/\!\bv), N(\bv\!/\!\bg)\}. \label{iq:Cii2}
\end{align}
As described on (\ref{eq:CIA}), similarly,
the left-hand side term in (\ref{eq:CIA2}) is given by the difference between the 3rd term
and the 1st term in (\ref{iq:Cii2}). %, that is
Therefore, if (\ref{eq:CIA2}) does not hold, then the 1st and the 3rd terms are equal.
In other words, $N(\bb_i) = \min\{N(\be\!/\!\bv), N(\bv\!/\!\bg)\}$ holds, so that $N(\bb_i)$ can be calculated.
Hence, $N(\bb_i)$ is not encoded if (\ref{eq:CIA2}) does not hold.
Therefore, in 3), $N(\bb_i)$ is encoded if both (\ref{eq:CIA}) and (\ref{eq:CIA2}) hold.

Let $I'(\be\!/\!\bv\!/\!\bg)$ be $\min(N(\be\!/\!\bv), N(\bv\!/\!\bg), N(\bv)\!-\!N(\be\!/\!\bv),$
$N(\bv)\!-\!N(\bv\!/\!\bg)) + 1$
where $\min(\cdot)$ is the left-hand term of (\ref{eq:CIA2}).
For encoding $N(\bb_i)$ by an entropy coding,
a probability is assigned to $N(\bb_i)$ as follows.
\begin{align}
\frac{1}{mn}\ \ &(\bb_i \in \mathcal{B}_1(\bp)), \label{eq:PP0}\\
\max\left(\frac{1}{I(\bb_i)}, \frac{1}{I'(\bb_i)}\right) \ \ &(\bb_i \in \mathcal{B}_2(\bp)),\label{eq:PP1}\\
\frac{|\mathcal{T}(\mathcal{B}(\bp), \bp, i)|}{|\mathcal{T}(\mathcal{B}(\bp), \bp, i\!-\!1)|}\ \ &(\bb_i \in\mathcal{B}_3(\bp)). \label{eq:PP2}
\end{align}
The assigned probabilities are encoded by an entropy coding such as an arithmetic coding.
For $\bp$,
the proposed algorithm outputs a following quartet
\begin{align}
(E(m), E(n), e(\bb_2, \bb_3, \dots, \bb_{|\mathcal{B}(\bp)|}), \epsilon(\text{rank($\bp$)})). \label{eq:Pcode}
\end{align}
In (\ref{eq:Pcode}), $E(m)$ and $E(n)$ represent encoded $m$ and $n$ by means of Elias integer code, respectively.
And rank($\bp$) represents an index for identifying $\bp$ in $[\bp]$
such as the rank of $\bp$ in $[\bp]$ with lexicographical order column-wisely.
Then, $\epsilon$(\text{rank($\bp$)}) represents an encoded rank($\bp$) by $\lceil \log_2 mn \rceil$ bits, and
$e(\bb_2, \bb_3, \dots, \bb_{|\mathcal{B}(\bp)|})$
represents a sequence of $N(\bb_i)~(2\leq i \leq |\mathcal{B}(\bp)|)$ which
are encoded by an entropy coding as described in Sec~\ref{sec:Conv}.

In the proposed algorithm, 
in (P-i), number of encoded $N(\bb_i)$ is $|\cX|\!-\!1$, that is a constant, 
while that in (C-i) is exponential with respect to $m$, that is $|\cX|^m\!-\!1$.
As for (P-ii), number of candidates $N(\bb_i)$ for encoding is polynomial order with respect to $m$ and $n$.
The reason is as follows.
As for 1), it is the same as (C-ii).
As for 2) and 3), since $\bv$ is a r-core, from the discussions on a c-core described in Sec.~\ref{sec:Conv}, 
the total number of candidates $N(\bb_i)$ for encoding is polynomial order with $m$ and $n$.
In other words, the set of all the candidates can be utilized instead of $\mathcal{B}(\bp)$ in (P-ii) in practice.
Similarly, note that $\mathcal{B}(\bp)$ is utilized for simplifying the explanation in this paper.
Hence, %as for computational time, 
for a 2D source $\bp$, 
the total number of output blocks of the proposed algorithm %can work in 
is polynomial with respect to $m$ and $n$ while
that of the conventional CSE %works in 
is exponential with respect to $m$.
\section{Evaluation of the Proposed Algorithm}
A general source $\mathbf{X}$ is defined as 
\[\mathbf{X}\!:=\!\{X^{[m, n]}\!=\!(X_{(1, 1)}^{<m, n>}, X_{(1, 2)}^{<m, n>}, \dots, X_{(m, n)}^{<m, n>})\}_{m\!=\!1, n\!=\!1}^{\infty, \infty}\]
where a random variable $X^{[m, n]}$ takes a value in the $m \times n$ Cartesian product $\cX^{[m, n]}$ of $\cX$~\cite{Han02}.
The probability distribution of a random variable $X^{[m, n]}$ is denoted by $P_{X^{[m, n]}}$.
For $\mathbf{X}$, the sup-entropy rate of $\mathbf{X}$ is defined as
\begin{align}
\hat{H}(\mathbf{X}):=\limsup_{m\rightarrow \infty, n \rightarrow \infty}\frac{1}{mn}H(X^{[m, n]}). \label{eq:Hx}
\end{align}
%For $1\leq k \leq m$ and $1\leq l \leq n$, let $\tilde{X}^{[n+k-1, m+l-1]}$ be $(X^{[m, n]}\!:\!X^{[m, l-1]})\!/\!(X^{[k-1, n]}\!:\!X^{[k-1, l-1]})$,
%and for $\bw \in \cX^{[k, l]}$,
%let $P^{[k, l]}(\bw\,|\,X^{[m, n]})$ be
%\begin{align}
%&\frac{|\{(i, j):\tilde{X}_{(i, j)}^{(i+k-1, j+l-1)}\!=\!\bw, 1\!\leq\!i\!\leq\!m, 1\!\leq\!j\!\leq\!n\}|}{mn}\nonumber\\
%&=\frac{N(\bw\,|\,X^{[m, n]})}{mn}\label{eq:Pt}
%\end{align}
%where %$(i, j)$ is a coordinate and 
%\tilde{X}_{(i, j)}^{(i+k-1, j+l-1)}$ is a subblock of $\tilde{X}^{[n+k-1, m+l-1]}$.

For $\bp$, let $\ell(\bp)$ be a codeword length of the proposed algorithm. 
Let $\ell_0(\bp)$ be the total codeword length of $E(m)$, $E(n)$, and $\epsilon$(\text{rank($\bp$)}) in (\ref{eq:Pcode}).
The codeword length of $e(\bb_2, \bb_3, \dots, \bb_{|\mathcal{B}(\bp)|})$
consists of three parts $\ell_{1}(\bp)$, $\ell_{2}(\bp)$, and $\ell_{3}(\bp)$ where
$\ell_{1}(\bp)$, $\ell_{2}(\bp)$, and $\ell_{3}(\bp)$ are  the total codeword length of ${N}(\bb_i)$
for $\bb_i \in \mathcal{B}_1(\bp)$, $\bb_i \in \mathcal{B}_2(\bp)$, and $\bb_i \in \mathcal{B}_3(\bp)$, respectively.
Here, $\ell(\bp) = \ell_0(\bp) + \ell_{1}(\bp)+\ell_{2}(\bp)+\ell_{3}(\bp)$.

Theorem~\ref{th:CADF} is one of our main results. To prove Theorem~\ref{th:CADF}, we show three lemmas.
Lemma~\ref{le:cmb} is a 2D version of Lemma~3~\cite{Yo11}, and the proofs of Lemmas~\ref{le:cmb} and~\ref{le:eqT} are omitted
in this paper.
\begin{theorem}\label{th:CADF}
For a general source $\mathbf{X}$,
\[\limsup_{m, n\rightarrow \infty}E\left[\frac{\ell(X^{[m,n]})}{mn}\right]=\hat{H}(\mathbf{X}).\]
\end{theorem}
\begin{lemma}\label{le:cmb}
For $\bp$, $1\!\leq\!k\!\leq\!m$, and $1\!\leq\!l\!\leq\!n$
\[\log_2|\mathcal{T}(\bp, k, l)| \leq -\frac{mn}{kl}\sum_{\bw \in \cX^{[k, l]}} \frac{N(\bw\,|\,\bp)}{mn}\log \frac{N(\bw\,|\,\bp)}{mn}. \]
\end{lemma}
\vspace{-1pt}
\begin{lemma}\label{le:eqT}
If $\bb_{i+1} \in \mathcal{B}(\bp)$ such that $|\bb_{i+1}|_c \geq 2$ does not satisfy (\ref{eq:CIA})
or such that $|\bb_{i+1}|_r \geq 2$ does not satisfy (\ref{eq:CIA2}), then
$\mathcal{T}(\mathcal{B}(\bp), \bp, i+1) = \mathcal{T}(\mathcal{B}(\bp), \bp, i).$
\end{lemma}
\begin{lemma}\label{le:Exp}
\begin{align*}
&\limsup_{m, n\rightarrow \infty}-\frac{1}{KL}\!\!\!\!\!\!\!\!\sum_{\bw\in\cX^{[K, L]}}\!\!\!\!\!\!\!E\left[\frac{N(\bw\,|\,X^{[m, n]})}{mn}\right]\log_2 E\left[\frac{N(\bw\,|\,X^{[m, n]})}{mn} \right]\\
&= \hat{H}(\mathbf{X}).
\end{align*}
\end{lemma}
\begin{proof}
For $\bw \in \cX^{[K, L]}$,
$P_{X^{[m, n]}}(\bw)$ can be written by
\begin{align*}
&E\left[\frac{|\{(i, j)\text{ s.t. }X_{(i, j)}^{(i\!+\!K\!-\!1, j\!+\!L\!-\!1)}\!=\!\bw,1\!\leq\!i\!\leq\!m', 1\!\leq\!j\!\leq\!n'\}|}{m'n'}\right]
\end{align*}
where $m'$ and $n'$ are $m\!-\!K\!+\!1$ and $n\!-\!L\!+\!1$, respectively, and $(i, j)$ is a coordinate.
For $\bp$, let $N'(\bw\,|\,\bp)$ be $|\{(i, j)\text{ s.t. }\bp_{(i, j)}^{(i\!+\!K\!-\!1, j\!+\!L\!-\!1)}\!=\!\bw, 1\!\leq\!i\!\leq\!m', 1\!\leq\!j\!\leq\!n'\}|$.
Moreover, $\frac{N(\bw\,|\,\bp)}{mn}$ can be written by 
$\left(\frac{N'(\bw\,|\,\bp)+\delta}{m'n'}\right)\left(\frac{m'n'}{mn}\right)$
where $0\!\le\!\delta\!\le\!(K\!-\!1)(n\!-\!L\!+\!1)\!+\!(L\!-\!1)m$ from~(\ref{eq:N}). 
Since $K$ and $L$ are respectively $\lfloor\sqrt{\log_{|\cX|}\log_{|\cX|} m}\rfloor$ and $\lfloor\sqrt{\log_{|\cX|}\log_{|\cX|} n}\rfloor$, 
$\frac{N(\bw|\bp)}{mn}$ converges to $\frac{N'(\bw|\bp)}{m'n'}$ as $m$ and $n$ go to infinity.
Since $E\left[\frac{N'(\bw|X^{[m, n]})}{m'n'}\right]=P_{X^{[m, n]}}(\bw)$,
\begin{align*}
&\limsup_{m, n\rightarrow \infty}-\frac{1}{KL}\!\!\!\!\!\!\!\!\sum_{\bw\in\cX^{[K, L]}}\!\!\!\!\!\!\!\!E\left[\frac{N(\bw\,|\,X^{[m, n]})}{mn}\right]\log_2 E\left[\frac{N(\bw\,|\,X^{[m, n]})}{mn} \right]\\
&= \limsup_{m, n\rightarrow \infty}-\frac{1}{KL}\!\!\!\!\!\!\!\!\sum_{\bw\in\cX^{[K, L]}}\!\!\!\!\!P_{X^{[m, n]}}(\bw)\log_2 P_{X^{[m, n]}}(\bw)\\
&= \limsup_{m, n\rightarrow \infty}\frac{H(X^{[K, L]})}{KL} = \hat{H}(\mathbf{X}).
\end{align*}
\end{proof}

\begin{proof}[(Proof of Theorem~\ref{th:CADF})]
As for $\ell_0(\bp)$, from the assumption, since $m\le n$, 
%\begin{align*}
$\ell_{0}(\bp)\!\leq2(\log_2 n\!+\!2\log_2\log_2 n\!+\!7)\!+\!\lceil \log_2 mn \rceil$
%\end{align*}
where $(\log_2 n\!+\!2\log_2\log_2 n\!+\!7)$ and $\lceil \log_2 mn \rceil$
are costs of Elias integer code for $n$ and $\epsilon$(\text{rank($\bp$)}), respectively. 
As for $\ell_{1}(\bp)$, the cost of $N(\bb_i)$ in (P-i) is $\lceil\log_2 mn\rceil $ bits from (\ref{eq:PP0}),
so that $\ell_{1}(\bp)\!\leq\!(|\cX|\!-\!1)\lceil\log_2 mn\rceil$. 
As for $\ell_{2}(\bp)$, since $I(\bb_i) \le mn$ and $I'(\bb_i) \le mn$, 
costs of $I(\bb_i)$ and $I'(\bb_i)$ are at most $\log_2 mn $ bits. 
Moreover, since $m \le n$ and $K \le L$, 
\begin{align*}
\ell_{2}(\bp)&\leq\!\sum_{h = 1}^K \!\sum_{w = 1}^{L}|\cX|^{wh}\log_2 mn \leq L^2|\cX|^{L^2}\log_2 mn\\
&\leq 2(\log_{|\cX|}\log_{|\cX|}n)(\log_{|\cX|} n)(\log_2 n).
\end{align*}
Therefore, 
\begin{align}
\lim_{m, n\rightarrow \infty} (\ell_{0}(\bp)+\ell_{1}(\bp)+\ell_{2}(\bp))/mn = 0. \label{eq:pT01}
\end{align}

As for $\ell_{3}(\bp)$, from (\ref{eq:PP2}), 
cost of $N(\bb_i)$ is $-\log_2 (|\mathcal{T}(\mathcal{B}(\bp), \bp, i)|/|\mathcal{T}(\mathcal{B}(\bp), \bp, i\!-\!1)|)$
bits. 

Cost of the next encoded $N(\bb_j)$ such that $N(\bb_i)$ has been encoded immediately before $N(\bb_j)$ is  
$-\log_2 (|\mathcal{T}(\mathcal{B}(\bp), \bp, j)|/|\mathcal{T}(\mathcal{B}(\bp), \bp, j\!-\!1)|)$.
From Lemma~\ref{le:eqT}, $|\mathcal{T}(\mathcal{B}(\bp), \bp, j\!-\!1)|\!=\!|\mathcal{T}(\mathcal{B}(\bp), \bp, i)|$.
Therefore, $N(\bb_j)$ can be written by 
$-\log_2(|\mathcal{T}(\mathcal{B}(\bp), \bp, j)|/|\mathcal{T}(\mathcal{B}(\bp), \bp, i)|)$,
Hence, the denominator $|\mathcal{T}(\mathcal{B}(\bp), \bp, i)|$ for $\bp_j$ is equal to 
the previous numerator $|\mathcal{T}(\mathcal{B}(\bp), \bp, i)|$ for $\bb_i$, so that they are canceled.
Moreover, since $|\mathcal{T}(\mathcal{B}(\bp), \bp, |\mathcal{B}(\bp)|)|\!=\!|[\bp]|\!=\!\!mn$,
\begin{align}
\ell_{3}(\bp) = \log_2 |\mathcal{T}(\mathcal{B}(\bp), \bp, S\!-\!1)| - \log_2 mn.\label{eq:C3}
\end{align}
where $S$ is the index of the first block $\bb_S \in \mathcal{B}_3(\bp)$ which is encoded by arithmetic coding.
From Lemma~\ref{le:eqT}, 
$|\mathcal{T}(\mathcal{B}(\bp), \bp, S\!-\!1)|\!=\!|\mathcal{T}(\bp, K, L)|$.
Therefore, 
\begin{align}
\ell_{3}(\bp) = \log_2 |\mathcal{T}(\bp, K, L)| - \log_2 mn.\label{eq:C4}
\end{align}
From (\ref{eq:C4}) and Lemma~\ref{le:cmb},
\begin{align}
&\ell_{3}(\bp) \leq -\frac{mn}{KL}\!\!\!\!\sum_{\bw \in \cX^{[K,L]}}\!\!\!\frac{N(\bw)}{mn}\log_2\frac{N(\bw)}{mn}- \log_2 mn.\label{eq:C5}
\end{align}
Therefore, 
\begin{align}
&E\left[\frac{\ell_{3}(X^{[m, n]})}{mn}\right]\!\leq\nonumber\\
&-\frac{1}{KL}\!\!\!\!\!\!\!\sum_{{\bw}\in \cX^{[K, L]}}\!\!\!\!\!\!\!E\left[\frac{N(\bw|X^{[m, n]})}{mn}\log_2\frac{N(\bw|X^{[m, n]})}{mn}\right]\!-\!\frac{\log_2 mn}{mn}.\nonumber
%\label{eq:E1}
\end{align}
From Jensen's inequality, $E[\frac{N(\bw|X^{[m, n]})}{mn}]E[\log_2\frac{N(\bw|X^{[m, n]})}{mn}] \le E[\frac{N(\bw|X^{[m, n]})}{mn}\log_2\frac{N(\bw|X^{[m, n]})}{mn}]$. Therefore, from Lemma~\ref{le:Exp},
\begin{align}
\limsup_{m, n\rightarrow \infty}E\left[\frac{\ell_{3}(X^{[m, n]})}{mn}\right]\!&\leq\!\hat{H}(\mathbf{X}).\label{eq:T02}
\end{align}
From (\ref{eq:pT01}) and (\ref{eq:T02}), 
\begin{align}
\limsup_{m, n\rightarrow \infty}E\left[\frac{\ell(X^{[m, n]})}{mn}\right]\!&\leq\!\hat{H}(\mathbf{X}).\label{eq:T03}
\end{align}
The proposed code is a prefix code, so that Kraft's inequality is satisfied. Therefore, 
$\limsup_{m, n\rightarrow \infty}E\left[\frac{\ell(X^{[m, n]})}{mn}\right]\!\ge\!\hat{H}(\mathbf{X})$.
\end{proof}
From Remark 1.7.3~\cite{Han02}, if $\mathbf{X}$ is a stationary source, $\hat{H}(\mathbf{X})$ can be expressed by
$H(\mathbf{X})(:=\lim_{m, n \rightarrow \infty} \frac{H(X^{[m, n]})}{mn})$, that is the entropy rate of $\mathbf{X}$.
Therefore, if $\mathbf{X}$ is a stationary source, the average codeword length of the proposed algorithm converges
to $H(\mathbf{X})$ as $m$ and $n$ go to infinity.

\section{Conclusion \label{conclusion}}
For reducing computational time, we proposed a new CSE for a 2D source which utilizes the flat torus of the source
while the conventional CSE utilizes the circular string of the source as a probabilistic model.
The total number of output blocks of the new CSE is polynomial
while that of the conventional CSE is exponential with respect to the source size. 
The new CSE encodes the source in block-by-block while the conventional CSE does in line-by-line. 
Moreover, we prove that an upper bound on the average codeword length of the proposed CSE converges to the sup-entropy rate for a general
source as size of the input source goes to infinity.
Furthermore, if a general source is a stationary source, then the length converges to the entropy rate of the source as
the size goes to infinity.
\bibliographystyle{ieeetr}

\begin{thebibliography}{10}

\bibitem{DB10}
D.~Dub\'e and V.~Beaudoin, ``Lossless data compression via substring
  enumeration,'' in {\em Proc. of the Data Compression Conference 2010},
  pp.~229--238, Mar. 2010.

\bibitem{Yo11}
H.~Yokoo, ``Asymptotic optimal lossless compression via the cse technique,'' in
  {\em Proc. of the Data Compression, Communications and Processing 2011},
  pp.~11--18, June 2011.

\bibitem{DY11}
D.~Dub\'e and H.~Yokoo, ``The universality and linearity of compression by
  substring enumeration,'' in {\em Proc. of the 2012 IEEE International
  Symposium on Information Theory}, pp.~1619--1623, Aug. 2011.

\bibitem{BD14}
D.~Dub\'e and V.~Beaudoin, ``Improving compression via substring enumeration by
  explicit phase awareness,'' in {\em Proc. of the Data Compression Conference
  2014}, pp.~26--28, Mar. 2014.

\bibitem{KYYK16}
S.~Kanai, H.~Yokoo, K.~Yamazaki, and H.~Kaneyasu, ``Efficient implementation
  and empirical evaluation of compression by substring enumeration,'' {\em
  IEICE Transactions on Fundamentals}, vol.~E99-A, no.~2, pp.~601--611, 2016.

\bibitem{BW94}
M.~Burrows and D.~Wheeler, ``A block-sorting lossless data compression
  algorithm,'' {\em SRC Research Report}, pp.~73--93, May 1994.

\bibitem{OM13}
T.~Ota and H.~Morita, ``On antidictionary coding based on compacted substring
  automaton,'' in {\em Proc. of the 2013 IEEE International Symposium on
  Information Theory}, pp.~1754--1758, July 2013.

\bibitem{CMRS00}
M.~Crochemore, F.~Mignosi, A.~Restivo, and S.~Salemi, ``Data compression using
  antidictionaries,'' in {\em Proc. of IEEE}, pp.~1756--1768, Nov. 2000.

\bibitem{OM141}
T.~Ota and H.~Morita, ``On a universal antidictionary coding for stationary
  ergodic sources with finite alphabet,'' in {\em Proc. of the 2014
  International Symposium on Information Theory and its Applications},
  pp.~294--298, October 2014.

\bibitem{IA15}
K.~Iwata and M.~Arimura, ``Lossless data compression via substring enumeration
  for $k$-th order markov sources with a finite alphabet,'' {\em IEICE
  Transactions on Fundamentals}, vol.~E99-A, no.~12, pp.~2130--2135, 2016.

\bibitem{CFMP16}
M.~Crochemore, G.~Fici, R.~Marcu\c{s}, and S.~Pissis, ``Linear-time sequence
  comparison using minimal absent words \& applications,'' in {\em Proc. of
  2016 Latin American Symposium on Theoretical Informatics}, pp.~334--346, Apr.
  2016.

\bibitem{Eli75}
P.~Elias, ``Universal codeword sets and representations of the integers,'' {\em
  IEEE Trans. Inform. Theory}, vol.~IT-21, no.~2, pp.~194--203, 1975.

\bibitem{MT95}
A.~Moffat and A.~Turpin, {\em Compression and coding algorithms}.
\newblock Kluwer Academic Publishers, 2002.

\bibitem{Han02}
T.~S. Han, {\em Information-spectrum methods in information theory}.
\newblock Springer-Verlag, 2002.

\end{thebibliography}

\end{document}